\documentclass[10pt,conference]{IEEEtran}
\usepackage[utf8]{inputenc}
\usepackage[english]{babel}
\usepackage{xcolor}
\usepackage{tikz}
\usetikzlibrary{quantikz}

\usepackage{multirow}

\usepackage{hyperref}
\usepackage{url}
\usepackage{booktabs}
\usepackage{amsfonts}
\usepackage{nicefrac}
\usepackage{microtype}

\usepackage{amsmath, amsfonts, amssymb}
\usepackage{amsthm}
\usepackage{mathtools}
\usepackage{graphicx}
\usepackage{cite}

\newcommand{\R}{\mathbb R}
\newcommand{\C}{\mathbb C}

\newcommand{\N}{\mathbb N}

\newcommand{\opt}{^{\star}}

\newcommand{\tran}{^{\top}}
\newcommand{\var}{\mathrm{var}}

\newcommand{\diag}{\mathrm{diag}}
\newcommand{\bd}{\mathrm{Beta}}

\usepackage{bbm}

\newcommand{\ev}[1]{\mathbb{E}\left[#1\right]}
\newcommand{\pr}[1]{\mathbb{P}\left(#1\right)}

\newtheorem{theorem}{Theorem}

\newtheorem{lemma}[theorem]{Lemma}

\newcommand{\partialk}{\frac{\partial}{\partial\theta_k}}
\renewcommand{\Re}{\mathsf{I\!Re}}
\newcommand{\tr}{\mathsf{Tr}}

\newcommand{\dagg}{^{\dagger}}

\renewcommand{\Re}{\mathsf{I\!Re}}
\renewcommand{\Im}{\mathsf{I\!Im}}

\newcommand{\Vct}{V_{\mathsf{ct}}}
\newcommand{\ct}{{\mathsf{ct}}}

\newcommand{\U}{\mathrm{U}}
\newcommand{\SU}{\mathrm{SU}}

\usepackage{algorithm,algpseudocode}

\newcommand{\leo}{\leq \mathcal{O}}
\newcommand{\eeo}{= \mathcal{O}}
\newcommand{\geo}{\ge \mathcal{O}}
\newcommand{\bigo}{\mathcal{O}}
\newcommand{\cng}{\gategroup[2,steps=1,style={dashed,rounded corners,fill=white, inner xsep=2pt},background]{}}

\begin{document}
\title{Sketching the Best Approximate Quantum Compiling Problem
}

\author{
\IEEEauthorblockN{Liam Madden$^1$, Albert Akhriev$^2$, Andrea Simonetto$^3$}
\vspace*{2mm}
\IEEEauthorblockA{$^1$ University of Colorado Boulder, USA. Email: liam.madden@colorado.edu}
\IEEEauthorblockA{{$^2$ IBM Quantum, IBM Research Europe, Dublin, Ireland}}
\IEEEauthorblockA{$^3$ UMA, ENSTA Paris, Institut Polytechnique de Paris, 91120 Palaiseau, France}
\vspace*{-5mm}}

\maketitle

\begin{abstract}
This paper considers the problem of quantum compilation from an optimization perspective by fixing a circuit structure of CNOTs and rotation gates then optimizing over the rotation angles. We solve the optimization problem classically and consider algorithmic tools to scale it to higher numbers of qubits. We investigate stochastic gradient descent and two sketch-and-solve algorithms. For all three algorithms, we compute the gradient efficiently using matrix-vector instead of matrix-matrix computations. Allowing for a runtime on the order of one hour, our implementation using either sketch-and-solve algorithm is able to compile 9 qubit, 27 CNOT circuits; 12 qubit, 24 CNOT circuits; and 15 qubit, 15 CNOT circuits. Without our algorithmic tools, standard optimization does not scale beyond 9 qubit, 9 CNOT circuits, and, beyond that, is theoretically dominated by barren plateaus.
\end{abstract}










\section{Introduction}

With the steady advances in quantum hardware and volume~\cite{jurcevic2020demonstration}, quantum computing is well on track to become widely adopted in science and technology in the near future. One of the core challenges to enable its use is the availability of a flexible and reliable quantum compiler, which can translate any target quantum circuit into a circuit that can be implemented on real hardware with gate set, connectivity, and length limitations. 

Several works have focused on how to efficiently map different gates into canonical (universal) gate sets up to an arbitrary accuracy, e.g.,~\cite{dawson2005,selinger2013nqubit,redu2008,redu2013}, or how to ``place'' the target circuit onto the real connectivity-limited hardware, e.g.~\cite{place2008,place2019c,tket,tan2020optimal,Giacomo2021}. In this paper, we are interested instead in best approximate quantum compilation, meaning finding a circuit that can be implemented in hardware, that is the closest as possible (with respect to a pertinent metric) to a desired (or target) circuit.  

Works for the best approximate quantum compilation problems appeared in the literature as~\cite{cincio2018,khatri2019,younis2020qfast,younis2021qfast,squander,rakyta2021approaching,our}. Here we focus primarily on the recent~\cite{khatri2019,our}, which formulate the problem as a mathematical optimization program over properly parameterized hardware-compatible circuits. In particular, in~\cite{our}, one defines a target circuit as a unitary matrix in $n$ qubits, $U$, and a parametric ansatz $\Vct(\theta)$ built upon allowed gates and interconnections, and solves (classically) the optimization problem:
\begin{equation}\label{theproblem}
    \min_{\theta}\, \frac{1}{2d}\|\Vct(\theta)-U\|_F^2,
\end{equation}
where $\|\cdot\|_F$ denotes the Frobenious norm and $d=2^n$. In~\cite{khatri2019}, the approach is qualitatively the same, even though the authors use a slightly different parametric ansatz and they show how to run the optimization on a quantum computer (albeit with considerable error due to quantum noise).

Since for near-term applications, the approximate quantum compilation problem will have to be solved classically, we focus here on {\bf algorithmic tools to boost its scaling} from the current $n=5$ (for random targets) in \cite{our} and $n=9$ (for short targets) in \cite{khatri2019}. In particular, we look at characterizing the landscape properties of~\eqref{theproblem} and we study computational techniques coming from stochastic gradient descent and sketching ~\cite{ghadimi2013stochastic,halko2011finding}, which involve fast and efficient computations.  

We remark that in many applications, e.g., the ones stemming from multi-body Hamiltonian simulations, even a small improvement in the number of qubits could be of considerable importance. Larger-scale circuits can then be treated in a hierarchical way, by dividing such circuits into smaller pieces, which will gain more flexibility if one is able to compile slightly larger pieces. In addition, the aim here is not to propose a method that scales to any $n$ (which may not be possible in general), but to improve upon existing techniques. Last but not least, let us remark that adding $1$ qubit in the quantum circuit, multiplies the size of the involved matrices by $4$.

\smallskip
{\bf Contributions.} Our stepping stone is our previous work~\cite{our}. There, we looked at target unitary matrices to be compiled in the space of special unitary matrices of dimension $d$, $\SU(d)$,  where $d=2^n$ and $n$ is the number of qubits. In particular, we considered random unitary target circuits which require a circuit length of $L\ge \frac{1}{4}\left(4^n-3n-1\right)$ in order to be exactly compiled. We solved the optimization problem~\eqref{theproblem} using Nesterov's method and were able to exactly compile circuits up to $n=5$. Here, we study how to go beyond that. 

First, we follow~\cite{khatri2019} in (1) restricting the parametric ansatz to have a special spin structure with a number of parameters that grows linearly\footnote{In particular, in~\cite{khatri2019} the authors consider a structure with an initial layer of $z$-rotations, a middle layer consisting of one cycle of the spin structure without rotations, and a final layer of $z$-rotations. In total, this parametric circuit has $2n$ parameters and they go up to $n=9$.} in the number of qubits $n$; (2) considering target matrices $U$ generated by taking random values of $\theta$ in the parametric ansatz we choose, i.e., matrices that have the same structure as the matrix $\Vct(\theta)$ (so we know that we can exactly compile them, if we are able to solve Problem~\eqref{theproblem} to optimality); (3) solving the optimization problem using the quasi-Newton method L-BFGS~\cite{byrd1995limited} whenever possible, since it is much faster than Nesterov's method\footnote{L-BFGS is not only faster, but it is also more conservative on the choice of parameter selection. In particular, we use default history size and there is no need to specify the learning rate as in the gradient descent case.}. 

Note that we restrict to the spin structure for simplicity and could easily consider other structures. Considering target matrices from the same structure is also for simplicity----so that we know the minimum of the objective is zero----but it makes sense in a context where we want to compile the target with high precision. Requiring that the number of parameters grows linearly in $n$ allows us to be in a computationally easier setting than the more complete problem we looked at in~\cite{our}, and therefore allows us to find what scale is achievable for such circuits. This could further guide how to compile circuits with significantly more parameters. With this in place,


$\bullet$ We explain that, based on \cite{mcclean2018barren} and \cite{cerezo2021cost}, Problem~\eqref{theproblem} has a barren plateaus property, which prevents standard optimization from scaling beyond $n \approx 9$;

$\bullet$ We modify Problem~\eqref{theproblem} by considering a stochastic variant thereof and we propose three algorithms (one based on stochastic gradient descent~\cite{ghadimi2013stochastic} and two on sketching~\cite{halko2011finding}). The sketching algorithms solve the carefully modified problem, which does not necessarily have barren plateaus, so that we are able to scale beyond $n = 9$ to $n=15$. We empirically show that the modification still allows us to optimize the original  Problem~\eqref{theproblem}. 

Our numerical results are the following:

$\bullet$ For $n=9$ and $L=3n=27$ CNOTs we have $3n+4L=135$ parameters. We find that 10\% sketching is sufficient to achieve a 25\% success rate running in around 30 minutes.

$\bullet$ For $n=12$ and $L=2n=24$ CNOTs we have $3n+4L=132$ parameters. We find that 5\% sketching is sufficient to achieve a 70\% success rate running in around 100 minutes.

$\bullet$ For $n=15$ and $L=n=15$ CNOTs we have $3n+4L=105$ parameters. We find that 0.3\% sketching is sufficient to achieve a 20\% success rate running in around 80 minutes.

Note that the success rate does not need to be close to 1, as it just indicates how many runs (possibly in parallel) we expect to take before finding a solution to Problem~\eqref{theproblem}.



As we successfully applied sketching in such a way so as to find solutions to Problem~\eqref{theproblem} with $n=15$ instead of just $n=9$, we see great merit in this approach and believe this paper to be the first stepping stone towards the use of (classical) sketching in quantum compiling.



\smallskip
{\bf Organization.} The remainder of the paper is organized as follows. Section~\ref{sec:formulation} describes the problem formulation, the choice of the ansatz, and the barren plateau property. In Section~\ref{sec:algos}, we present and analyze three stochastic-based algorithms, while in Section~\ref{sec:comp}, we focus on describing how to compute gradients efficiently (gradients are the cornerstone of our formulation, so their efficient computation is key). Finally, Section~\ref{sec:results} showcases the numerical results and we close in Section~\ref{sec:conclusions}. The Appendix contains proofs and more mathematical intuition on the choices made.  

\smallskip
{\bf Notation.} Notation is wherever possible standard. The number of qubits is indicated with $n$, which corresponds to circuits represented by unitary matrices in $d = 2^n$ dimension, indicated by $\U(d)$. We also indicate with $\SU(d)$ the special unitary matrices of dimension $d\times d$. For a matrix $Q$, $Q\dagg$ represents its Hermitian transpose.  

Further, to avoid confusion with our sketching matrix $\Omega$ and variables $\theta$, we use the following asymptotic notation: when we write $f \leo(\cdot)$, we denote the usual big-O notation, when we write $f \eeo(\cdot)$ or just $\bigo(\cdot)$ we mean the usual $f = \Theta(\cdot)$, and finally $f \geo(\cdot)$ denotes the usual $f \geq \Omega(\cdot)$.

\section{Problem formulation}\label{sec:formulation}

\begin{figure*}
\centering
\begin{tikzpicture}
\node[scale=.65] {
\begin{quantikz}
& \gate{R_z} & \gate{R_y} & \gate{R_z} &\ctrl{1}\cng&\qw&\ctrl{1}\cng&\qw&\ctrl{1}\cng&\qw&\ctrl{1}\cng&\qw&\qw\\
& \gate{R_z} & \gate{R_y} & \gate{R_z} &\targ{}&\ctrl{1}\cng&\targ{}&\ctrl{1}\cng&\targ{}&\ctrl{1}\cng&\targ{}&\ctrl{1}\cng&\qw\\
& \gate{R_z} & \gate{R_y} & \gate{R_z} &\ctrl{1}\cng&\targ{}&\ctrl{1}\cng&\targ{}&\ctrl{1}\cng&\targ{}&\ctrl{1}\cng&\targ{}&\qw\\
& \gate{R_z} & \gate{R_y} & \gate{R_z} &\targ{}&\qw&\targ{}&\qw&\targ{}&\qw&\targ{}&\qw&\qw&\rstick{}
\end{quantikz}
\quad with \quad 
\begin{quantikz}
 & \ctrl{1} \cng & \qw  \\  & \targ{} & \qw 
\end{quantikz}
$\equiv$
\begin{quantikz}
 & \ctrl{1} & \gate{R_y} & \gate{R_z} & \qw \\  & \targ{} & \gate{R_y} & \gate{R_x} & \qw & \rstick{.}
\end{quantikz}};
\end{tikzpicture}

\caption{Our spin ansatz for $n=4$ qubits and $L=12$ CNOTs. The CNOTs boxed with dashes are CNOT ``units'' as shown on the right-hand side. $\Vct(\theta)$ corresponds to a quantum circuit with the shown structure with its $3n+4L=60$ parameters (one for each rotation gate) given by $\theta\in\R^{3n+4L}$.}
\label{fig:spin}
\end{figure*}
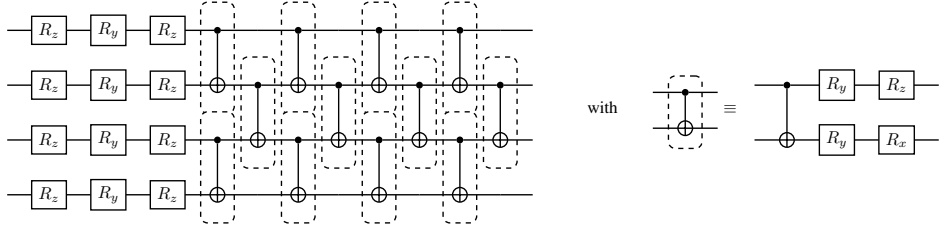


We are interested in compiling a quantum circuit, which we formalize as finding the ``best'' circuit representation in terms of an ordered gate sequence of a target unitary matrix $U\in \U(d)$, with some additional hardware constraints. In particular, we look at representations that could be constrained in terms of hardware connectivity, as well as {circuit length}, and we choose a gate basis in terms of CNOT and rotation gates. The latter choice is motivated by an implementation in the Qiskit software package~\cite{Qiskit2019}. We recall that the combination of CNOT and rotation gates is universal in $\SU(d)$ and therefore it does not limit compilation~\cite{nielsenbook}.

To properly define what we mean by ``best'' circuit representation, we define the metric as the Frobenius norm between the unitary matrix of the compiled circuit $V$ and the target unitary matrix $U$, i.e., $\|V - U\|_{\mathrm{F}}$. This choice is motivated by mathematical programming considerations, and it is related to other formulations that appear in the literature~\cite{our}. 

We are now ready to formalize the approximate quantum compiling problem as follows. 

\emph{Given a target special unitary matrix $U \in \SU(2^n)$ and a set of constraints, in terms of connectivity and {length}, find the closest special unitary matrix $V \in \mathcal{V} \subseteq \SU(2^n)$, where $\mathcal{V} $ represents the set of special unitary matrices that can be realized with rotations and CNOT gates alone and satisfy both connectivity and {length} constraints, by solving the following mathematical program: }
\begin{equation}\label{eq:aqcp}
\textbf{(AQCP)}\qquad \min_{V \in \mathcal{V} \subseteq SU(2^n)} \, f(V):=\frac{1}{2d}\|V - U\|_{\mathrm{F}}^2.
\end{equation}

We call~\eqref{eq:aqcp} the approximate quantum compiling (master) problem (AQCP). A solution of the problem is an optimal $V$ indicated as $V\opt$, along with an ordered set of gate operations that respect the constraints.

\subsection{A parametric circuit (ansatz)}

Problem~\eqref{eq:aqcp} is a very hard combinatorial problem to solve. In~\cite{our, khatri2019} and other works, one parametrizes $V$ with a fixed topology of CNOTs and rotation gates of free rotation angles (the parameters), and indicates with $\Vct(\theta)$ the resulting implementable ansatz. With this in place, one can rewrite~\eqref{eq:aqcp} as \eqref{theproblem}. The resulting problem is a continuous albeit nonconvex optimization problem. The choice of ansatz (and specifically its length) determines the quality of the obtained solution in terms of approximation error. In~\cite{our}, we have presented several ansatz design possibilities, while in \cite{khatri2019} a spin ansatz is proposed. As we argue in~\cite{our}, the topology of the ansatz is (on average) not a determining factor in the approximation error (as long as it satisfies some minimal properties such as being connected), while its length is. 

Here, we focus on the spin ansatz reported in Figure~\ref{fig:spin}, which is quite flexible in terms of connectivity (only a line connectivity is required), but our qualitative results will not change if one were to select a different ansatz.

\subsection{Equivalent problem formulation}

We now re-look at Problem~\eqref{theproblem} under a slightly different lens to facilitate our algorithm design and theoretical analysis. To do so, we introduce a lemma to transform the norm objective into an inner product formulation. This transformation is useful for two main reasons. First, the inner product formulation will help us compute the derivatives of the objective functions in a more efficient way, as shown in Section~\ref{sec:comp}. Second, its form naturally motivates a sketching approach, as we will see in Section~\ref{subsec:rla}. The proof of the lemma is in Section~\ref{app:equiv} of the Appendix. 

\begin{lemma}
\label{lem:equiv}
Let complex matrix $Q\in\C^{d\times m}$ have orthonormal columns. Then, we have that
\begin{equation*}
    \frac{1}{2m}\|(V-U)QQ\dagg\|_F^2 = 1-\frac{1}{m}\Re\langle VQ,UQ\rangle\in[0,2].
\end{equation*}
\qed
\end{lemma}

Note that for the special case $m=d$ and $Q=I$, we get
\begin{align*}
    \frac{1}{2d}\|V-U\|_F^2=1-\frac{1}{d}\Re\langle V,U\rangle\in [0,2],
\end{align*}
which can be used to transform Problem~\eqref{theproblem} as the maximization of $\Re\langle \Vct(\theta),U\rangle/d\in[-1,1]$.



\subsection{Preliminary results: Concentration and Barren Plateaus}
\label{subsec:bp}

Before presenting our main algorithms, let us focus briefly on the notion of barren plateaus, which is quite important in quantum computing and which further motivates our approach. Loosely speaking, we can show that the objective of  Problem~\eqref{theproblem} satisfies the barren plateaus property, which makes it prohibitively hard to optimize. In particular, standard algorithms cannot scale beyond $n \approx 9$ for Problem~\eqref{theproblem}.


Details and intuition on barren plateaus are given in Appendix~\ref{app:concentration}. The main takeaway is that the work \cite{cerezo2021cost} shows that for parameterized circuits such as $\Vct$ and ``global'' objective functions such as that of Problem~\eqref{theproblem}, if $L\geo (1)$ then there are barren plateaus. On the other hand, \cite{mcclean2018barren} also shows that for ``local'' objective functions, if $L\geo (n)$ then there are barren plateaus, while if $L\leo (\log(n))$ then there aren't. We have a global objective and $L\eeo (n)$, and so {\bf Problem~\eqref{theproblem} has barren plateaus}. 

This further motivate us to modify the objective and consider a \emph{sketched objective} with sketching dimension $m\eeo (\log(d))$. While we do not prove that the sketched objective does not have barren plateaus for $L\leo (n)$, we do provide numerical evidence that this is the case. Let us look now at how to build such a modified objective.


\section{Algorithms}\label{sec:algos}

We are now ready to propose our main algorithms. In particular, we will discuss three of them, each rooted in different formulation or approximation of Problem~\eqref{theproblem}.

We remark that the state of the art for Problem~\eqref{theproblem} with $L\eeo (n)$ is to go up to $n=9$ (using the quasi-Newton method L-BFGS~\cite{byrd1995limited}).
In order to scale to higher numbers of qubits, there are three things to consider: (1) the complexity of the gradient computation, (2) the number of gradient calls required to reach convergence to a local minimum, and (3) the suitability of the local minimum. We will focus on decreasing the complexity of the gradient computation while keeping our eye to the other two considerations. See Table~\ref{tab:alg} for the three algorithms that we consider.

We remark that even though the gradient can already be computed efficiently via a quantum circuit, as done in \cite{khatri2019}, the barren plateaus issue presented in Section~\ref{subsec:bp} prevents \cite{khatri2019} from going beyond $n=9$. The barren plateaus problem is tied up with our second consideration.

\begin{table}
    \centering
     \caption{Algorithms for Problem~\eqref{theproblem}}
    \label{tab:alg}
    \begin{tabular}{cl}
    \toprule
    Algorithm & Description \\ \toprule
    SGD     &  Stochastic gradient descent: Sample a new mini-batch at \\ & each iteration\\
    \midrule
    S\&S-1     &  Sketch-and-solve 1: Sample a single mini-batch at the start\\
    \midrule
    S\&S-2     &  Sketch-and-solve 2: Use QR decomposition of sketched\\ 
    & error matrix to project problem, taking multiple epochs\\
    & of a fixed number of iterations\\
    \bottomrule
    \end{tabular}
\end{table}


Stochastic gradient descent (SGD) with batch-size $m$ divides the complexity of the gradient computation by $d/m$ but increases the number of gradient calls required to reach convergence. Additionally, there is evidence in the literature on neural networks that SGD may find local minima that are flat rather than sharp \cite{keskar2017large}. Ultimately, SGD converges too slowly for it to scale beyond $n=9$. However, determining the theoretical trade-off between the gradient computation and the number of gradient calls is illustrative for the two sketch-and-solve methods. 

The two sketch-and-solve methods, S\&S-1 and S\&S-2, minimize a sketched objective. With sketching dimension $m$, they divide the gradient computation by $d/m$. Moreover, they do not increase the number of gradient calls required to reach convergence, nor do they create further spurious local minima. In fact, they seem to have the effect of smoothing the optimization landscape and so making it easier to navigate. However, they solve sketched problems, not the true problem. Fortunately, the solutions found by S\&S-1 and S\&S-2 tend to be solutions to the true problem as well. Thus, they successfully scale us beyond $n=9$ (to $n=15$).

In the following, we will make these considerations more formal. We let $\Omega$ denote a complex matrix with real and imaginary components sampled independently from $N(0,1)$.




\subsection{Stochastic Gradient Descent}
\label{subsec:sgd}

The first algorithm is SGD. By re-writing~\eqref{theproblem} as a stochastic optimization problem, we can find a solution using SGD and so decrease the cost of computing the gradient. This is common in machine learning where datasets are so large that the gradient cannot be computed all at once. The downside of this approach is that SGD requires more gradient calls to converge than gradient descent (GD) does. In this section, we will take the necessary steps to determine theoretical guarantees on the convergence rate of SGD for our problem.

First, using the results of \cite{kania2015short,kania2015trace}, we re-write~\eqref{theproblem} as a stochastic optimization problem as follows:
\begin{multline}
\label{eq:stoch}
\min_{\theta} \frac{1}{2d} \|\Vct(\theta)-F\|_F^2 \underbrace{\equiv}_{\textrm{Lemma~\ref{lem:equiv}}}  \min_{\theta} -\frac{1}{d}\Re\langle \Vct(\theta),U\rangle
   \equiv \\ \min_{\theta} \ev{-\Re\langle \Vct(\theta)x,Ux\rangle}
\end{multline}
where $\ev{\cdot}$ represents the expectation taken with respect to the random vector $x\in \C^d$ sampled uniformly from the unit sphere.


Then taking $T$ iterations of SGD with batch-size $m$ and step-size sequence $(\eta_t)$ for $t\in \mathbb{N}$ gives the following algorithm.

\begin{algorithm}[H]
\begin{algorithmic}[1]
\Require $U$, ansatz $V_\ct(\cdot), (\eta_t),T,m$
\State Randomly initialize $\theta_0$
\For{$t=0,\ldots,T-1$}
\State Sample $\Omega\in \C^{d\times m}$
\State Compute $X$ by normalizing the columns of $\Omega$
\State $g\leftarrow -\nabla \Re\langle \Vct(\theta_t)X,UX\rangle/m$
\State $\theta_{t+1}\leftarrow \theta_t-\eta_tg_t$
\EndFor
\State \Return $\theta_T$
\end{algorithmic}
\caption{Stochastic gradient descent}
\label{algo:sgd}
\end{algorithm}


We now use results in SGD analysis to derive the theoretical properties of Algorithm~\ref{algo:sgd}. 

First, we use that the objective function of our problem has a $\rho$-Lipschitz continuous gradient by Theorem 2 of \cite{our}, with a bounded constant $\rho>0$ that can be estimated. If we were able to prove a bounded variance property, that there exists $\sigma^2<\infty$ such that
\begin{align*}
    \ev{\bigg\|\frac{\nabla \Re\langle \Vct(\theta)X,UX\rangle}{m}-\frac{\nabla \Re\langle \Vct(\theta),U\rangle}{d}\bigg\|^2}<\sigma^2~\forall \theta,
\end{align*}
then, for an appropriately chosen step-size sequence (depending on $\rho$), SGD would be guaranteed to compute $(\theta_t)$ such that
\begin{align*}
    \ev{\min_{t\in[T]}\bigg\|\frac{\nabla \Re\langle \Vct(\theta_t),U\rangle}{d}\bigg\|^2} \leo\left(\frac{1}{T}+\frac{\sigma}{\sqrt{T}}\right)
\end{align*}
See \cite{ghadimi2013stochastic} for the details.

By contrast, if we instead apply the standard gradient descent method (GD) with constant step-size $\eta=1/\rho$, then, using Eq.~(2.1.9) of~\cite{Nesterov}, it is not difficult to derive that
\begin{align*}
    \ev{\min_{t\in[T]}\bigg\|\frac{\nabla \Re\langle \Vct(\theta_t),U\rangle}{d}\bigg\|^2} \leo \left(\frac{\rho}{T}\right).
\end{align*}

Note that both of these bounds are on the norm of the gradient and so measure convergence to a first-order stationary point (a point where the gradient is zero).

While GD has better asymptotic properties in terms of the iteration count ($1/T$ vs. $1/\sqrt{T}$), one iteration of GD costs $d/m \gg 1$ times as much as one iteration of SGD (since we perform matrix-vector computations for maximum efficiency as explained in Section~\ref{sec:comp}).

Let us look at the complexity to reach a certain error. Assume $\rho$ is large enough that we can ignore the asymptotic $1/T$ term for SGD. Then in order to reach $\epsilon$-tolerance, SGD has worst-case complexity $\geo(m\sigma^2/\epsilon^2)$ while GD has worst-case complexity $\geo(d\rho/\epsilon)$. In order to compare these, we need to derive $\rho$ and $\sigma^2$. We already proved $\rho\leo(nd)$ in Theorem 2 of \cite{our}, but this bound may be pessimistic. In the following, we derive $\sigma^2\leo(n/(md))$. Thus, if we set $\epsilon\eeo(1/d)$, then $\epsilon$ goes to zero exponentially fast with $n$ and, if $\rho\geo (n)$, SGD reaches $\epsilon$-tolerance $\geo(d)$ times faster than GD.

\medskip
\noindent {\bf Bounded variance property}

We will show now that the bounded variance property holds with a constant that decreases exponentially fast with respect to increasing the number of qubits, i.e., we derive $\sigma^2\leo(n/(md))$, as claimed. To do so, we need the following supporting lemma.

\begin{lemma}
\label{lem:ratio}
Given a vector $a\in\R^n$ with components $a_i$ for $i \in[n]$, define the random variable
\begin{align*}
    Y(a) = \frac{a_1X_1^2+\cdots +a_nX_n^2}{X_1^2+\cdots+X_n^2}
\end{align*}
where $X_1,\ldots,X_n$ are independent standard normal random variables. Then the variance of $Y(a)$ is bounded as
\begin{align*}
    \var(Y(a)) \leq  \frac{4(n-1)}{n^2(n+2)}\sum_{i=1}^n a_i^2.
\end{align*}
\end{lemma}

Lemma~\ref{lem:ratio} allows us to characterize the variance of the stochastic objective function.

\begin{lemma}
\label{lem:unitbatch}
Let $U\in\U(d)$. Let the vector $x$ be a uniformly sampled complex unit vector. Then the variance of $\Re\langle Vx,Ux\rangle$ can be upper bounded as
\begin{align*}
   \sup_{V\in\U(d)}\var\left(\Re\langle Vx,Ux\rangle\right)\leq \frac{4(d-1)}{d(d+2)}.
\end{align*}
\end{lemma}
\begin{proof}
There exist unitary matrices $Q$ and $\Lambda=\diag(\lambda)$ such that $V\dagg U=Q\dagg \Lambda Q$. So, $\langle Vx,Ux\rangle = \langle Qx,\Lambda Q x\rangle$. Since $x$ is uniform and $Q$ is unitary, $Qx$ is uniform. Thus, $\Re\langle Vx,Ux\rangle$ is distributed the same as
\begin{align*}
    \Re\langle x,\Lambda x\rangle = \Re(\lambda_1)|x_1|^2+\cdots+\Re(\lambda_d)|x_d|^2,
\end{align*}
which is distributed the same as $Y(\Re(\lambda))$ as defined in Lemma~\ref{lem:ratio}, from which the result follows since $\Re(\lambda)\in[-1,1]^d$.
\end{proof}

In order to say something about the variance of the gradient rather than objective function, we apply Lemma~\ref{lem:unitbatch} to to get the following theorem about $p$-dimensional random vectors where each of the (possibly dependent) components has the same form as the random variable in Lemma~\ref{lem:unitbatch}.

\begin{theorem}
\label{thm:gradientnoise}
Let $U\in\U(d)$. Let $x_1,\ldots,x_m$ be uniformly sampled complex unit vectors. Then,
\begin{align*}
    &\sup_{V_i\in\U(d)}\ev{\bigg\|\left(\frac{1}{m}\sum_{k=1}^m\Re\langle V_ix_k,Ux_k\rangle-\frac{1}{d}\Re\langle V_i,U\rangle\right)_{i=1}^p\bigg\|^2}\\
    &\hspace{4cm}\leq \frac{4(d-1)p}{md(d+2)}.
\end{align*}
\end{theorem}
\begin{proof}
Observe, the left-hand side (LHS) is equivalent to
\begin{align*}
    &\sup_{V_i\in\U(d)}\ev{\sum_{i=1}^p\bigg|\frac{1}{m}\sum_{k=1}^m\Re\langle V_ix_k,Ux_k\rangle-\frac{1}{d}\Re\langle V_i,U\rangle\bigg|^2}\\
    &= \sum_{i=1}^p \sup_{V_i\in\U(d)}\var\left(\frac{1}{m}\sum_{k=1}^m\Re\langle V_ix_k,Ux_k\rangle\right)\\
    &= \frac{1}{m^2}\sum_{i=1}^p \sup_{V_i\in\U(d)}\var\left(\Re\langle V_ix_k,Ux_k\rangle\right) \leq \text{RHS}
\end{align*}
where the final inequality follows from Lemma~\ref{lem:unitbatch}.
\end{proof}

To apply Theorem~\ref{thm:gradientnoise}, note that $\nabla \Re\langle \Vct(\theta),U\rangle=(\Re\langle V_{\ct,i}(\theta),U\rangle/2)_{i=1}^{3n+4L}$ where the $V_{\ct,i}(\theta)$ can be seen in Section 5 of \cite{our}. Thus, we have that mini-batch stochastic gradient descent with batch-size $m$ applied to this problem satisfies the bounded variance property with
\begin{align*}
    \sigma^2 = \frac{(d-1)(3n+4L)}{md(d+2)}\eeo\left(\frac{n}{md}\right).
\end{align*}
The fact that the noise vanishes exponentially fast with increasing $n$ is why SGD can theoretically outperform GD. However, SGD still takes too long to converge and does not scale beyond $n=9$. Still, exponentially vanishing noise suggests that we could sketch the problem and still find a solution to the original problem. We see how to do that next.

\subsection{Sketch-and-solve}

The second algorithm we present is S\&S-1. Instead of sampling a new mini-batch at each iteration, we sample a single mini-batch, $X$, at the start and then solve the following sketched problem to completion:
\begin{equation}
\label{eq:sketch}
\min_{\theta}-\frac{1}{m}\Re\langle \Vct(\theta)X,UX\rangle.
\end{equation}
We describe it in the following algorithm.



\begin{algorithm}[H]
\begin{algorithmic}[1]
\Require $U$, ansatz $\Vct(\cdot), m$
\State Randomly initialize $\theta_0$
\State Sample $\Omega\in \C^{d\times m}$
\State Compute $X$ by normalizing the columns of $\Omega$
\State Apply L-BFGS to Problem~\eqref{eq:sketch} to get $\theta$
\State \Return $\theta$
\end{algorithmic}
\caption{Sketch-and-solve 1 (S\&S-1)}
\label{algo:ss1}
\end{algorithm}

The algorithm S\&S-1 converges to a local minimum of the sketched objective much faster than our basic method of applying L-BFGS to Problem~\eqref{theproblem}. Furthermore, our numerical experiments suggest that for $L\eeo(n)$ and $m\eeo(n)\eeo(\log(d))$, the sketched and original objectives are sufficiently close, that is
\begin{align}\label{approx}
    \frac{\Re\langle \Vct(\theta)X,UX\rangle}{m}\approx\frac{\Re\langle \Vct(\theta),U\rangle}{d}\qquad \forall \theta,
\end{align}
so that S\&S-1 does converge \emph{practically} to solutions of the Problem~\eqref{theproblem}, entailing an exponential decrease in computational complexity and allowing us to scale beyond $n=9$ (to $n=15$).

For the interested reader, in Appendix~\ref{app.ss1}, we show some preliminary theoretical results for~\eqref{approx}.

\subsection{Sketch-and-solve with randomized linear algebra}
\label{subsec:rla}

In the previous section, we set $X\in\C^{d\times m}$ by normalizing the columns of $\Omega$. We explore here another idea stemming from randomized linear algebra. In particular, one could project onto the subspace of the largest eigenvectors of $\Vct(\theta_0)-U$. That way one would eliminate the largest discrepancies first.

The idea is to compute the following skinny QR decomposition~\cite[Ch. 5.2]{golub2013matrix}:
\begin{align*}
    QR = (\Vct(\theta_0)\dagg-U\dagg)\Omega
\end{align*}
where $Q\in\C^{d\times m}$ has orthonormal columns, so $Q\dagg Q=I_m$ but $QQ\dagg \neq I_d$, and $R\in\C^{m\times m}$ is upper triangular. Then $QQ\dagg (\Vct(\theta_0)\dagg-U\dagg)\approx \Vct(\theta_0)\dagg-U\dagg$ for sufficiently large $m$ \cite{halko2011finding}. Using $Q$ and Lemma~\ref{lem:equiv}, we get a new sketched problem:
\begin{equation}
\label{eq:rla}
\min_{\theta}-\frac{1}{m}\Re\langle \Vct(\theta)Q,UQ\rangle.
\end{equation}

Intuitively, as we iteratively solve Problem~\eqref{eq:rla}, the largest discrepancies between $\Vct(\theta_0)$ and $U$ are eliminated first, eventually rendering $Q$ ineffective. Thus, it makes sense to restart every $T$ iterations for a total of $k$ epochs. This is formalized in the following algorithm.



\begin{algorithm}[H]
\begin{algorithmic}[1]
\Require $U$, ansatz $\Vct(\cdot), m,T,k$
\State Randomly initialize $\theta_0$
\For{$i=0,\ldots,k-1$}
\State Sample $\Omega\in \C^{d\times m}$
\State $QR\leftarrow (\Vct(\theta_i)\dagg-U\dagg)\Omega$
\State Starting at $\theta_i$, apply $T$ iterations of L-BFGS to Problem~\eqref{eq:rla} to get $\theta_{i+1}$
\EndFor
\State \Return $\theta_k$
\end{algorithmic}
\caption{Sketch-and-solve 2 (S\&S-2)}
\label{algo:ss2}
\end{algorithm}

\begin{figure*}
    \centering
    \includegraphics[width=\columnwidth]{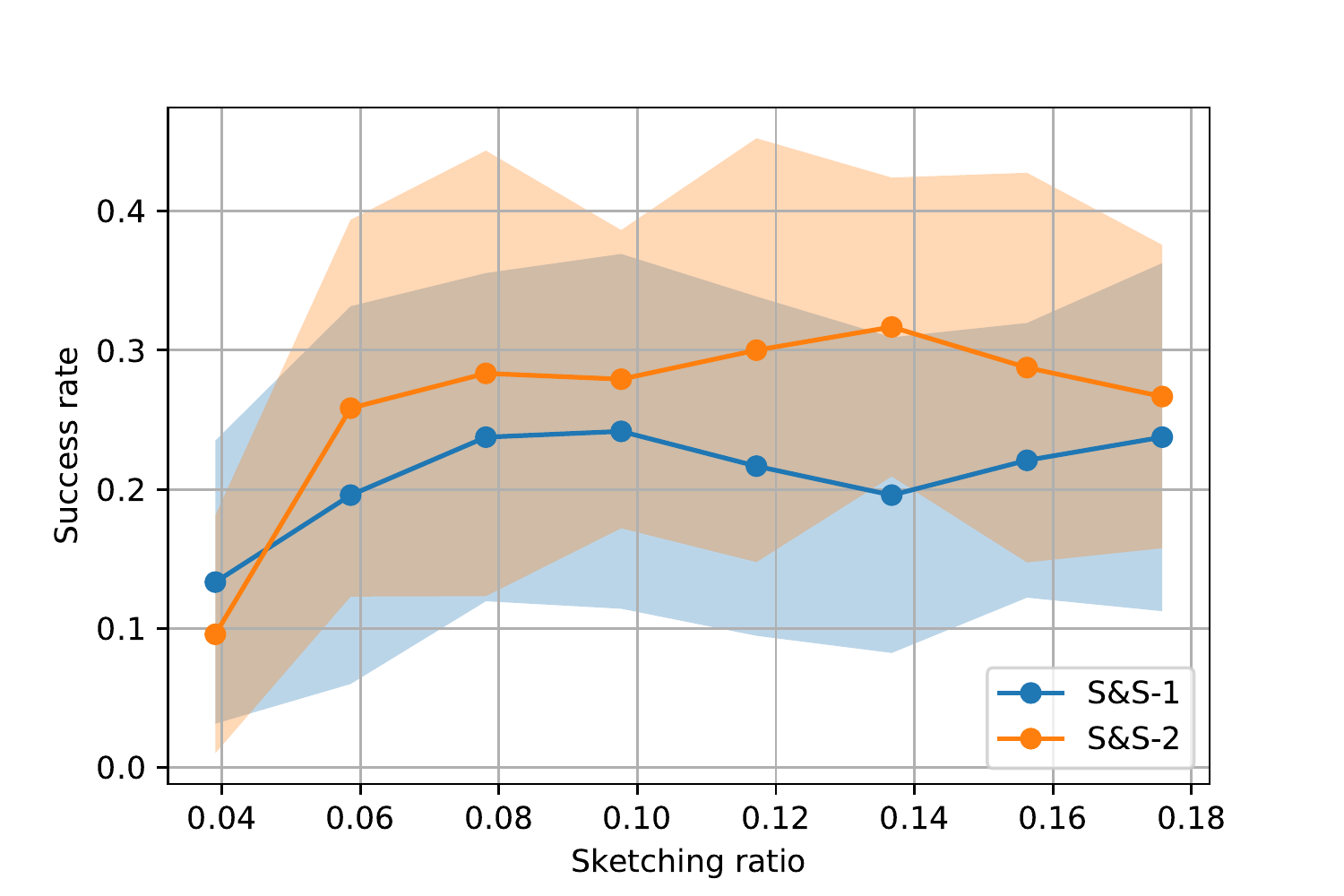}
    \includegraphics[width=\columnwidth]{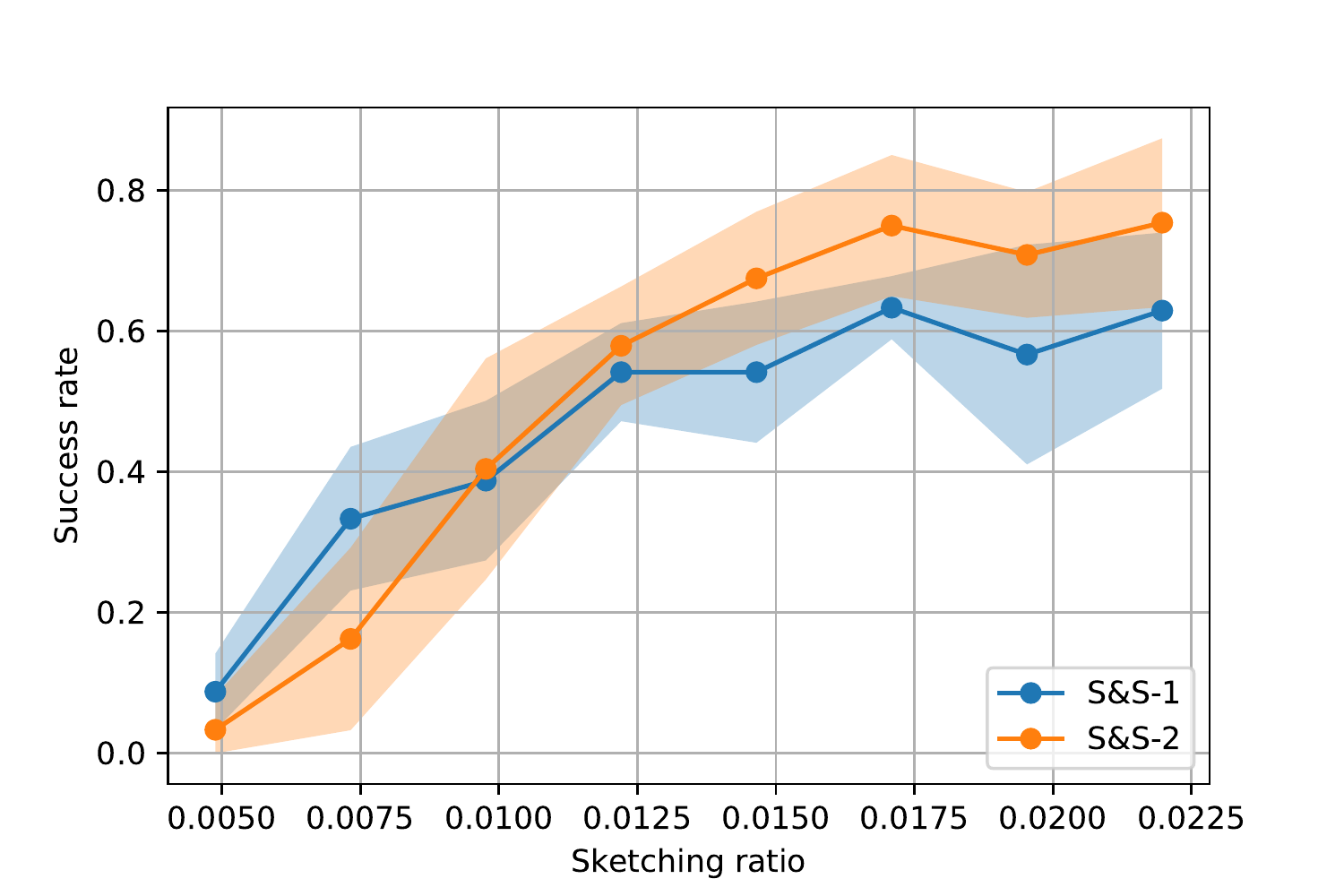}
    \caption{Success rate vs. sketching ratio ($m/d$) for (left) S\&S-1 and S\&S-2, both with 3 epochs, where $n=9$ and $L=27$ and (right) S\&S-1 and S\&S-2, both with 3 epochs, where $n=12$ and $L=24$. The continuous line is the average, while the shaded area is one standard deviation. }
    \label{fig:success}
\end{figure*}

\begin{figure*}
    \centering
    \includegraphics[width=\columnwidth]{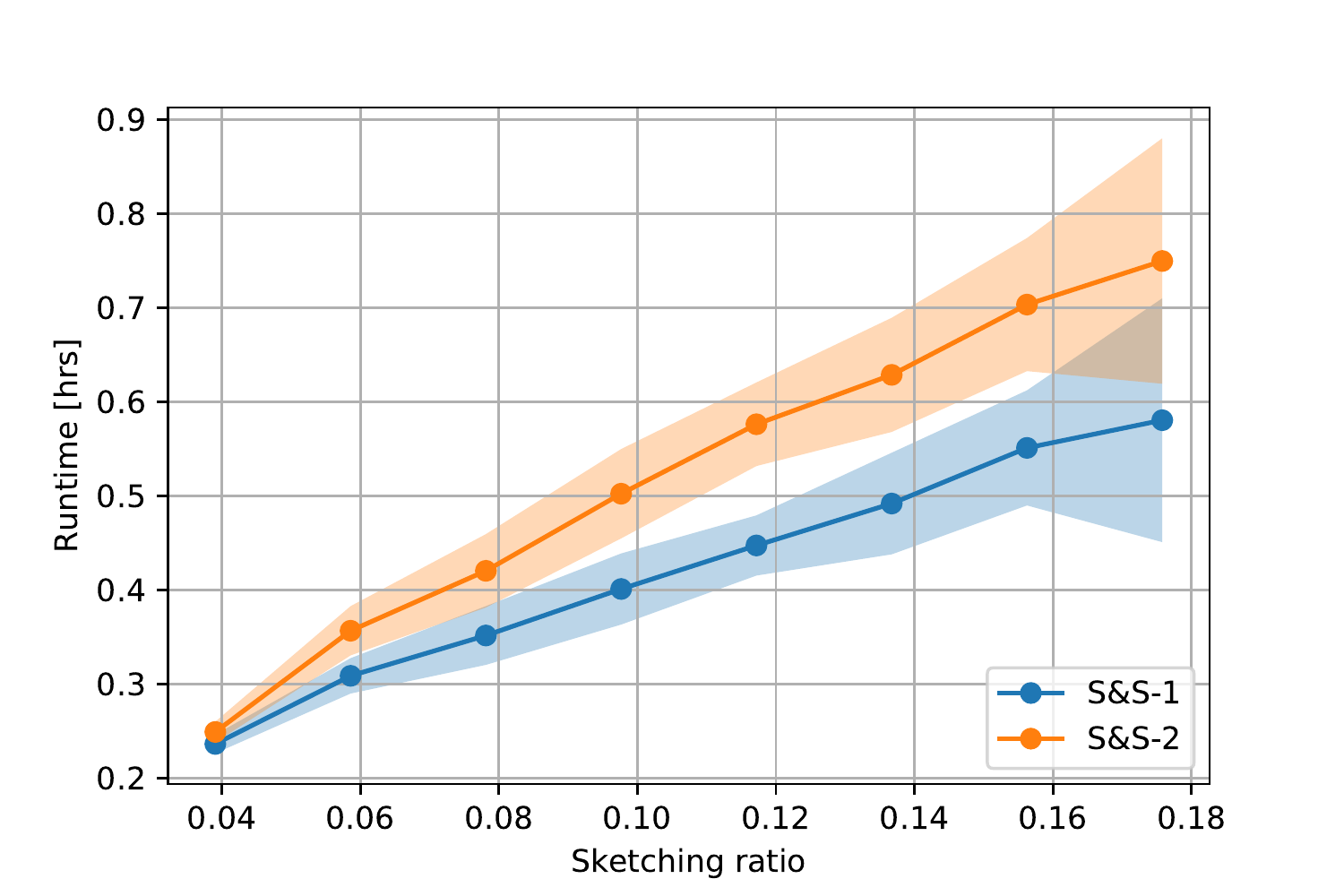}
    \includegraphics[width=\columnwidth]{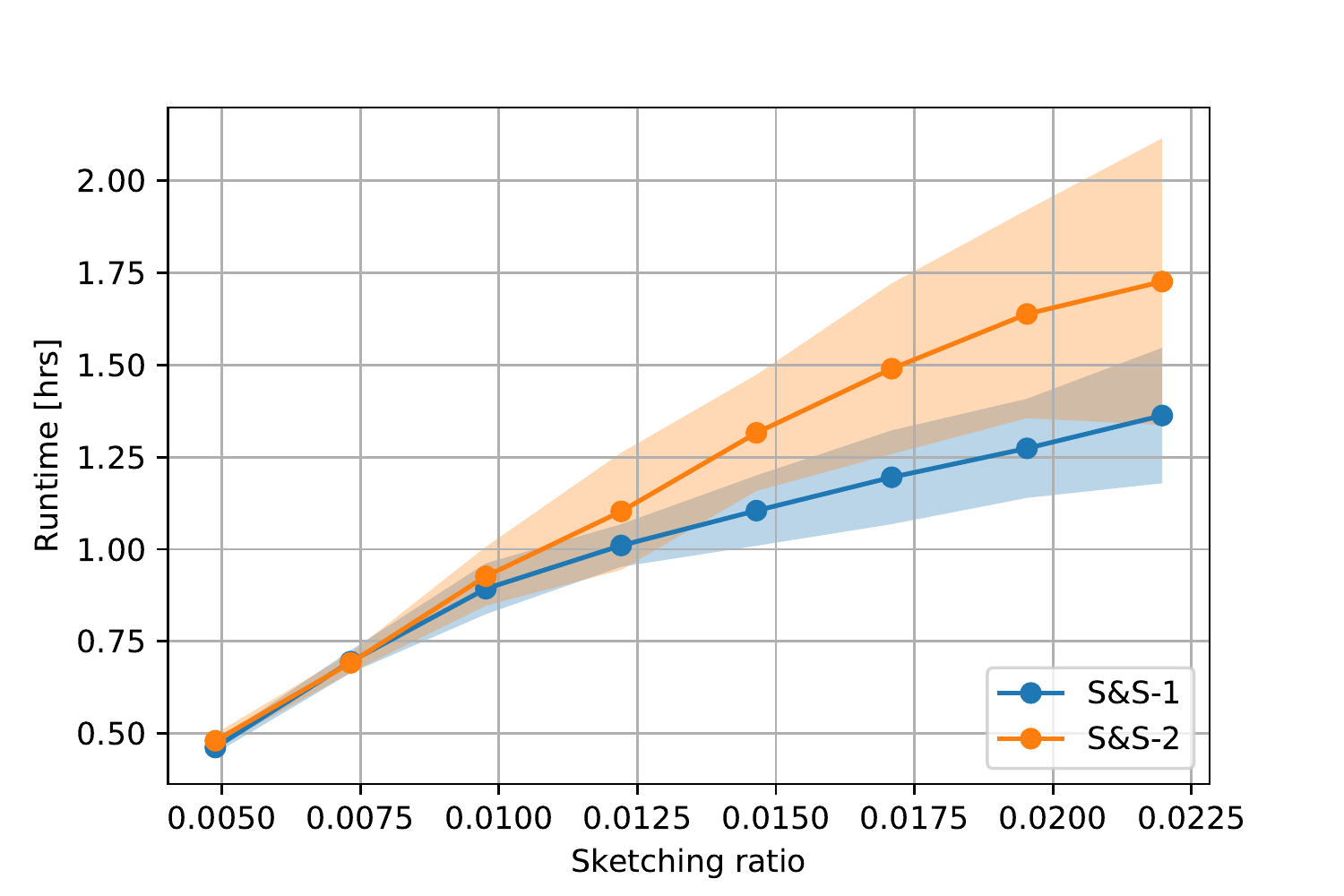}
    \caption{Runtime vs. sketching ratio ($m/d$) for (left) S\&S-1 and S\&S-2, both with 3 epochs, where $n=9$ and $L=27$ and (right) S\&S-1 and S\&S-2, both with 3 epochs, where $n=12$ and $L=24$. The continuous line is the average, while the shaded area is one standard deviation.}
    \label{fig:time}
\end{figure*}


While the theoretical characterization of S\&S-2 is more complex and left for future research, an analogy with~\cite{halko2011finding} suggests that we set the sample size to $m\geo(3n+4L)$.

\section{Efficient computation}
\label{sec:comp}

All the proposed methods hinge on gradient computations (including L-BFGS), and in particular we need to classically compute $\nabla \Re\langle \Vct(\theta)X,UX\rangle = \nabla \Re\langle \Vct(\theta)X,\Vct(\theta_U)X\rangle$ as efficiently as possible (using that $U = \Vct(\theta_U)$ for some $\theta_U$).

Done naively, this will take $3n+4L$ times as long as computing $\Re\langle \Vct(\theta)X,\Vct(\theta_U)X\rangle$. However, there is a way to reduce the complexity by storing intermediate computations in an automatic differentiation fashion. So, first we will discuss how to compute the objective $\Re\langle \Vct(\theta)X,\Vct(\theta_U)X\rangle$ efficiently, then we will discuss how to compute $\nabla \Re\langle \Vct(\theta)X,\Vct(\theta_U)X\rangle$ efficiently.

To compute $\Re\langle \Vct(\theta)X,\Vct(\theta_U)X\rangle$ with matrix-matrix multiplications takes $\bigo(d^2m)$ time and $\bigo(d^2)$ space (not including the more costly construction of the matrix $\Vct(\theta)$ itself). We will show how to improve this complexity by multiplying by each column of $X$ separately and so replacing each matrix-matrix multiplication with $m$ optimized matrix-vector multiplications.

The single-qubit gate $v$
applied to the $q$th qubit is the matrix $V=I_{2^{q-1}}\otimes v\otimes I_{2^{n-q}}$. Let $P$ permute the qubits so that the $q$th qubit goes to the $n$th qubit place. Then
\begin{align*}
    V = P\tran V' P \coloneqq P\tran \begin{pmatrix}
    v &&0\\
    &\ddots &\\
    0&& v
    \end{pmatrix} P.
\end{align*}
$V'$ is a block-diagonal matrix with 2 by 2 blocks. To multiply $V'$ times a vector $y$, we can reshape $y$ column-wise to the shape $(2,2^{n-1})$ and multiply $v$ from the left. Thus, the cost of computing $Vx$ is the cost of one $(2,2)(2,2^{n-1})$ matrix multiplication and two permutations of a vector in $\R^d$ (the reshaping cost gets absorbed into the permutation cost).

Similarly, to apply a CNOT with control $q_1$ and target $q_2$ to $x$, we can use a permutation $P$ to move the $q_1$th qubit to the $(n-1)$th qubit place and the $q_2$th to the $n$th. Then $V'$ is block-diagonal with blocks
\begin{align*}
    v = \begin{pmatrix}
    1&0&0&0\\
    0&1&0&0\\
    0&0&0&1\\
    0&0&1&0
    \end{pmatrix}.
\end{align*}
Now to apply $V'$ times a vector $y$, we can reshape $y$ column-wise to the shape $(4,2^{n-2})$ and multiply $v$ from the left, which ends up being the same as swapping the 3rd and 4th rows. Thus, we have derived how to compute CNOT \textit{as} a permutation.


Note that the permutation $P$ only depends on gate placement and can be computed once for every qubit position and pair of qubit positions. We do not construct the matrix $P$ but only store a permutation array of size $d$ in memory. The cost of applying it is $O(d)$ (though it is not a cache-friendly operation).

So, to multiply a single gate times a vector costs $\bigo(d)$ in both space and time. Since there are $\bigo(n+L)$ gates, we can compute $\Vct(\theta)x$ (or $V_{\ct,i}(\theta)x$) in $\bigo((n+L)d)$ time and $\bigo(d)$ space. Since there are $m$ vectors, we can compute $\Re\langle \Vct(\theta)X,\Vct(\theta_U)X\rangle$ in $\bigo((n+L)dm)$ time and $\bigo(dm)$ space. If $L=\bigo(n)$, then the time complexity becomes $\bigo(ndm)$.

To see how to compute $\nabla \Re\langle \Vct(\theta)X,\Vct(\theta_U)X\rangle$ efficiently, consider $F(\theta;x,y)=\langle V_p(\theta_p)\cdots V_1(\theta_1)x,y\rangle$ where $V_k$ is either a rotation gate, or a rotation gate and a CNOT (e.g. we can group the CNOT on the right-hand side of Figure~\ref{fig:spin} with the y-rotation on the control qubit). Let the Pauli matrix corresponding to $V_k$ be $\sigma_k\in\{x,y,z\}$. Thus, $\partialk V_k(\theta_k)=-\frac{i}{2}\sigma_k V_k(\theta_k)$ and so
\begin{align*}
    &\partialk F(\theta;x,y)\\
    &= \bigg\langle V_p(\theta_p)\cdots -\frac{i}{2}\sigma_k V_k(\theta_k)\cdots V_1(\theta_1)x,y\bigg\rangle\\
    &= \frac{i}{2}\Big\langle\sigma_k V_k(\theta_k)\cdots V_1(\theta_1)x,V_{k+1}(\theta_{k+1})\dagg\cdots V_p(\theta_p)\dagg y\Big\rangle.
\end{align*}
So, set $w_0=x$ and $z_0=V_1(\theta_1)\dagg\cdots V_p(\theta_p)\dagg y$ and, for $k\in[p]$,
\begin{align*}
    &w_k= V_k(\theta_k)w_{k-1}, \qquad z_k = V_k(\theta_k)z_{k-1}.
\end{align*}
Then
\begin{align*}
    \partialk F(\theta;x,y) &= \frac{i}{2}\langle\sigma_kw_k,z_k\rangle.
\end{align*}
Thus, the complexity of computing the gradient is on the same order as the complexity of computing the objective.

\section{Results}
\label{sec:results}

To showcase our proposed algorithms, we focus here on three circuit structures: a 9 qubit circuit structure with $L=3n=27$ CNOTs, and so $3n+4L=135$ parameters; a 12 qubit circuit structure with $L=2n=24$ CNOTs, and so $3n+4L=132$ parameters; and a 15 qubit circuit structure with $L=n=15$ CNOTs, and so $3n+4L=105$ parameters. 

First, we ran SGD and found that it was not able to find a solution for any of the circuit structures because the convergence rate is so slow. Thus, we only consider S\&S-1 and S\&S-2 in this section.

From the proposed sketching algorithms, given a target $U=\Vct(\theta_U)$ and a sketching dimension $m$, there is a probability that a run of S\&S-1 or S\&S-2 finds a solution to Problem~\eqref{theproblem}. Given a sketching dimension $m$, we would like to understand this success probability over the set of $\Vct(\theta_U)$. In particular, we would like to know (1) when the minimum success probability is non-zero and (2) what the distribution of the success probability is when $\theta_U$ is sampled from the uniform distribution. To answer these questions, we look at the following experiment.

We consider different sketching dimensions, $m$, in the set $\{20,30,40,50,60,70,80,90\}$, and we randomly sample $10$ different target circuits. For each random target, we run S\&S-1, or S\&S-2, $24$ times. For both methods, we do three epochs: that is, we randomly initialize then sketch-and-solve three times, using the previous solution as initialization for the latter two sketch-and-solves.\footnote{While multiple epochs are necessary for S\&S-2, since only the largest discrepancies are eliminated with the first epoch, we do multiple epochs of S\&S-1 as well so that the runtimes are comparable.} If a run results in a fidelity of $0.999$ or greater, then the run is considered a success. The fidelity is averaged over the Haar distribution and so, given matrices $U$ and $V$, is
\begin{align*}
    \Bar{F}(U,V)=\frac{1+\frac{1}{d}|\langle V,U\rangle|^2}{d+1}\in[0,1]
\end{align*}
(see Remark 1 of \cite{our}). We compute the success rate for each of the $10$ targets and then compute the sample mean and standard deviation. We plot the success rates against the sketching ratios in Figure~\ref{fig:success} and the runtimes against the sketching ratios in Figure~\ref{fig:time} for $n=9$ and $n=12$. The $n=15$ results are in Figure~\ref{fig:n15}.

\begin{figure*}
    \centering
    \includegraphics[width=\columnwidth]{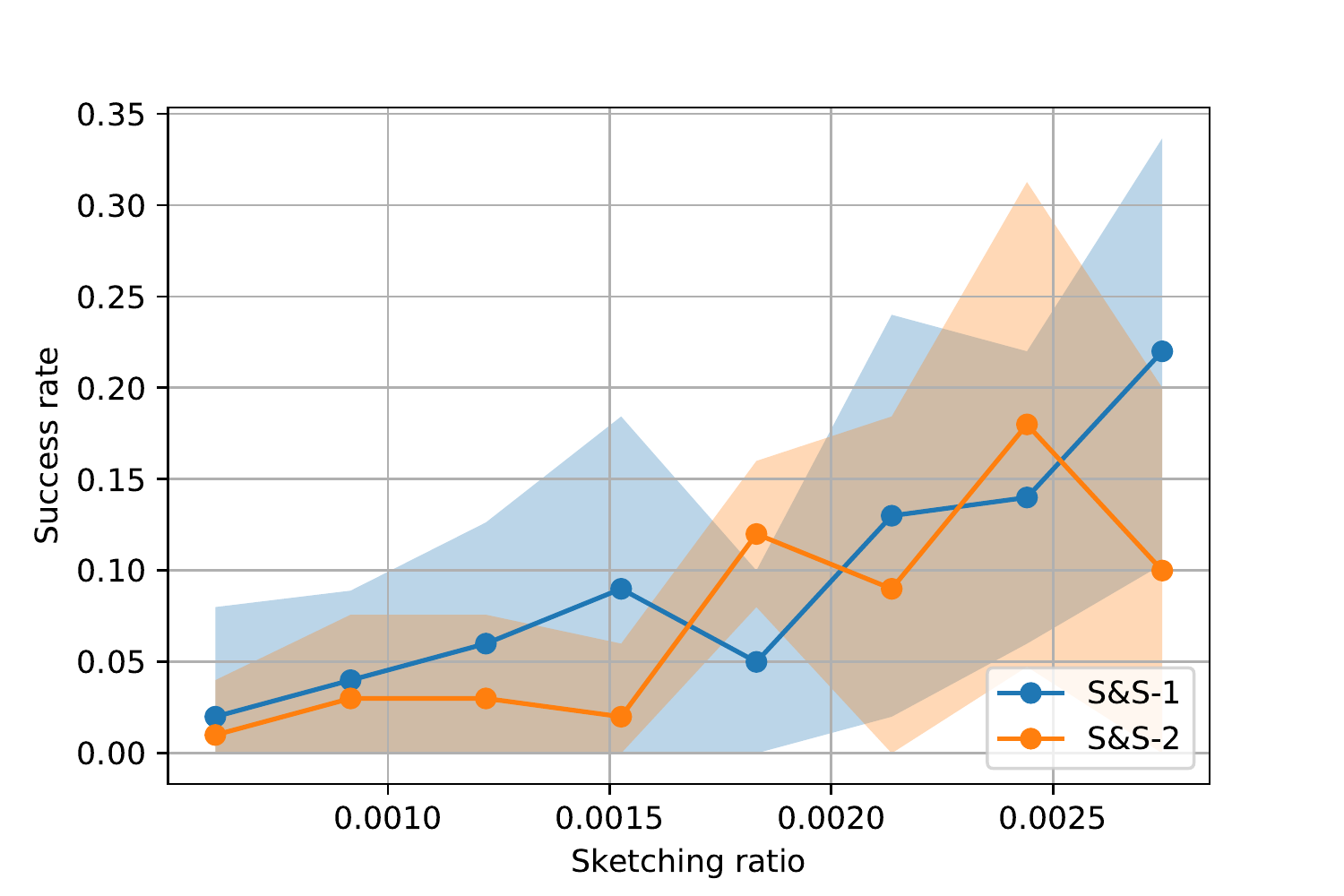}
    \includegraphics[width=\columnwidth]{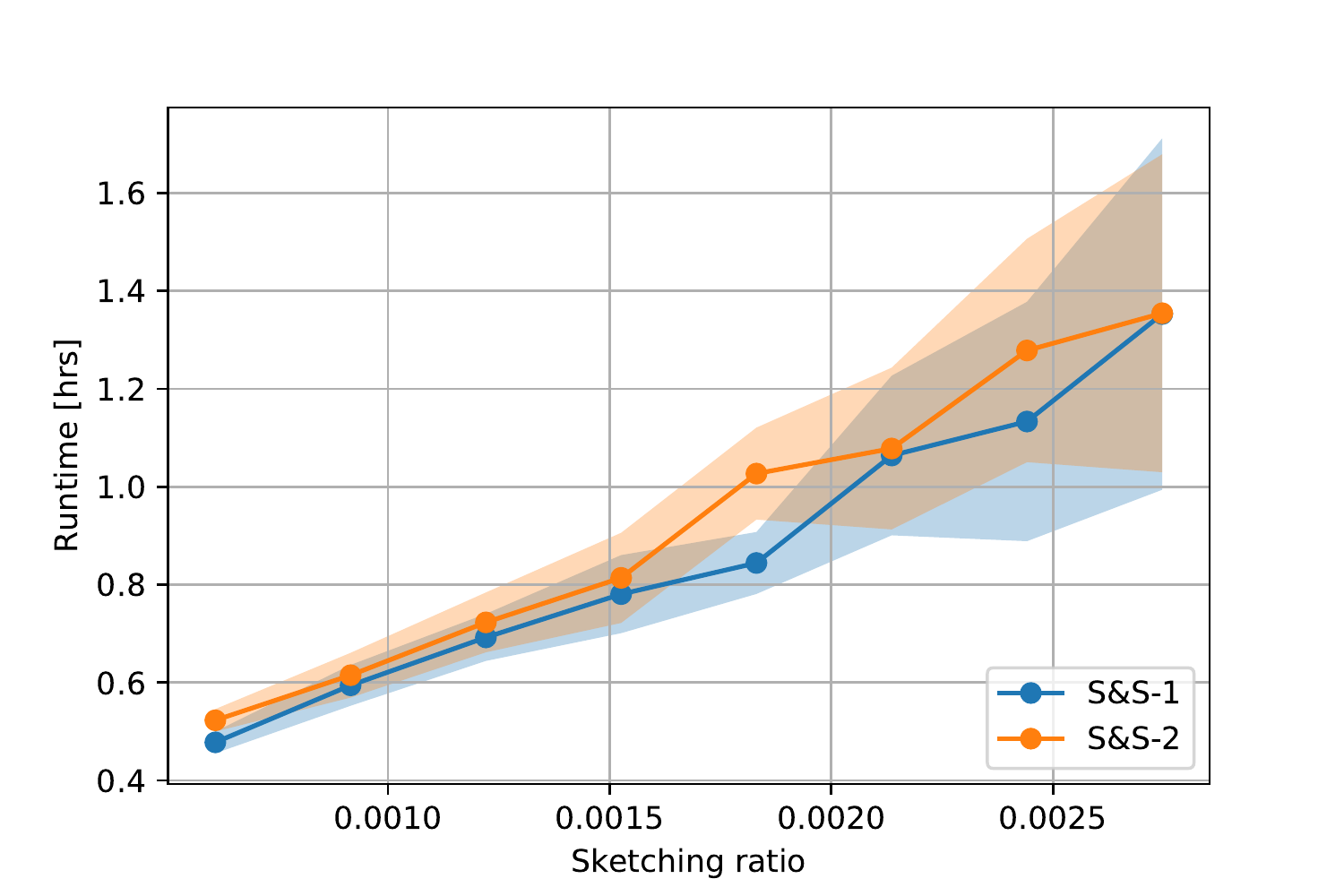}
    \caption{Success rate (left) and runtime (right) vs. sketching ratio ($m/d$) for S\&S-2 with 1 epochs, where $n=15$ and $L=15$. The continuous line is the average, while the shaded area is one standard deviation.}
    \label{fig:n15}
\end{figure*}


The success rate depends on both initialization and sketching. Thus, as we increase the sketching dimension, we do not expect the success rate to go to 1, but rather to go to the (unknown) initialization success rate. In choosing the parameters of the experiments, we observed that the initialization success rate becomes very small as $L$ is increased to $4n$ and beyond. This suggests that for higher numbers of CNOTs, it may be necessary to iteratively solve (similar to the block Kaczmarz method~\cite{NEEDELL2015322}, as in the sequential optimization procedure of \cite{squander,rakyta2021approaching}) or update (e.g. block coordinate descent) blocks of $3n$ parameters at a time. Finally, note that taking multiple runs is not an issue and, in fact, can be parallelized. 

For $n=9$, the success rate reaches around 25\% for 10\% sketching for both S\&S-1 and S\&S-2, though it is slightly higher for S\&S-2. At 10\% sketching, both S\&S-1 and S\&S-2 run in around 30 minutes, though S\&S-2 takes slightly longer. At 4\% sketching, S\&S-1 was able to compile every target in 24 tries (i.e. the minimum success rate over the targets was non-zero). At 6\% sketching, S\&S-2 was able to compile every target in 24 tries.

For $n=12$, the success rate reaches around 60\%, for S\&S-1, and 70\%, for S\&S-2, for 2\% sketching. At 2\% sketching, the runtime is around 80 minutes for S\&S-1 and 100 minutes for S\&S-2. At 0.7\% sketching, S\&S-1 and S\&S-2 were able to compile every target in 24 tries.

For $n=15$ we only did one epoch for both methods due to computation time. We also only did 10 trials per random target. The success rate reaches around 20\%, for S\&S-1, and 15\%, for S\&S-2, for 0.3\% sketching. At 0.3\% sketching, both S\&S-1 and S\&S-2 run in around 80 minutes. At 0.3\% sketching, S\&S-1 was able to compile every target in 10 tries. At 0.2\% sketching, S\&S-2 was able to compile every target in 10 tries.

Thus, for all three circuit structures---$n=9$, 12, and 15---the two sketch-and-solve methods find solutions to Problem~\eqref{theproblem} in a reasonable amount of time (on the order of an hour).




\section{Conclusions}
\label{sec:conclusions}

In this paper, we considered the problem of quantum compilation from an optimization perspective by fixing the circuit structure and optimizing over rotation angles. We were unable to scale beyond 9 qubits using standard optimization likely due to barren plateaus. So, we considered three different optimization problems, each suggesting its own algorithm. One algorithm was stochastic gradient descent and two were sketch-and-solve methods. For all three algorithms (and for the original problem) we explained how to efficiently compute the gradient. Our results were quite positive, showing that the sketch-and-solve algorithms can compile (in around an hour) 9 qubit, 27 CNOT circuits; 12 qubit, 24 CNOT circuits; and 15 qubit, 15 CNOT circuits. One observation is that the success rate with respect to random initialization becomes very small when there are 4 or more CNOTs per qubit. We leave it as a future research direction to consider block solvers with blocks of 3 CNOTs per qubit. 

\appendix

\subsection{Proof of Lemma~\ref{lem:equiv}}
\label{app:equiv}

Observe,
\begin{align*}
    &\|(V-U)QQ\dagg\|_F^2 =\tr\left[QQ\dagg(V\dagg-U\dagg)(V-U)QQ\dagg\right]\\
    &\hspace{1cm}=\tr\left[(V\dagg-U\dagg)(V-U)QQ\dagg QQ\dagg\right]\\
    &\hspace{1cm}= \tr\left[(V\dagg-U\dagg)(V-U)QQ\dagg\right]\\
    &\hspace{1cm}= 2\tr\left[QQ\dagg\right]-\tr\left[V\dagg UQQ\dagg\right]-\tr\left[U\dagg VQQ\dagg\right]\\
    &\hspace{1cm}=2m-2\Re\tr\left[V\dagg UQQ\dagg\right]\\
    &\hspace{1cm}=2m-2\Re\tr\left[Q\dagg V\dagg UQ\right]\!=\! 2m-2\Re\langle VQ,UQ\rangle.\\
\end{align*}

\subsection{Barren plateaus}
\label{app:concentration}

{\bf Barren plateaus.} Consider an objective function over a parameterized quantum circuit. A barren plateau is a region of the domain such that the gradients in that region concentrate around zero exponentially fast with respect to increasing $n$. If the volume of the domain that is not a barren plateau also goes to zero exponentially fast with respect to increasing $n$, then the objective is said to satisfy the barren plateaus property. The implications are twofold. First, the computational complexity for an optimization algorithm to escape a barren plateau increases exponentially with respect to increasing $n$. Second, the probability of randomly initializing outside of a barren plateau goes to zero exponentially fast with respect to increasing $n$.
This issue has been vastly treated recently in quantum computing~\cite{cerezo2021cost,hadfield2019quantum,lee2018generalized,mcclean2018barren,grant2019initialization,zhou2020quantum,volkoff2021large,skolik2021layerwise,cerezo2021higher}.
Roughly speaking, the problem with barren plateaus is that if an objective satisfies the property, then it would be very hard to optimize, and exponentially more so increasing the number of qubits. 

{\bf Haar distributed assumption.} To give a sense of the issue \emph{for our problem}, we include the following lemma. The lemma, applied to our problem, says that if $\Vct(\theta)$ is Haar distributed for uniformly random $\theta$, then our objective satisfies the barren plateaus property. We know that $\Vct(\theta)$ is not Haar distributed in general, since the image of $\Vct$ is much smaller than $\SU(d)$, but the lemma is still important to give an understanding on the main challenges. Keep in mind that the lemma is not an if and only if result, and we will discuss extensions to non-Haar distributions after the lemma. We indicate with $\pr{\cdot}$ the probability of a given event. 

\begin{lemma}
\label{lem:concentration}
Let $d\in\N$ and $U\in\U(d)$. Let $V$ be sampled uniformly from $\U(d)$. Then, for all $t>0$,
\begin{align*}
    \pr{|\langle V,U\rangle|>t}\leq \frac{1}{t^2}.
\end{align*}
\end{lemma}
\begin{proof}
First,
$
    \langle V,U\rangle = \sum_{j=1}^d\sum_{k=1}^d \overline{V_{jk}}U_{jk}, 
$
so
\begin{align*}
    |\langle V,U\rangle|^2 &= \Big(\sum_{j=1}^d\sum_{k=1}^d \Re\left(\overline{V_{jk}}U_{jk}\right)\Big)^2\\ & \hspace*{2cm} +\Big(\sum_{j=1}^d\sum_{k=1}^d \Im\left(\overline{V_{jk}}U_{jk}\right)\Big)^2.
\end{align*}
Developing the products and sums,
\begin{align*}
    |\langle V,U\rangle|^2 &= \sum_{j=1}^d\sum_{k=1}^d \sum_{\ell=1}^d\sum_{m=1}^d\Re\left(\overline{V_{jk}}U_{jk}\right)\Re\left(\overline{V_{\ell m}}U_{\ell m}\right)\\
    &\hspace{2.8cm}+\Im\left(\overline{V_{jk}}U_{jk}\right)\Im\left(\overline{V_{\ell m}}U_{\ell m}\right)\\
    &= \sum_{j=1}^d\sum_{k=1}^d \sum_{\ell=1}^d\sum_{m=1}^d\Re\left(\overline{V_{jk}}U_{jk}\overline{\overline{V_{\ell m}}U_{\ell m}}\right)\\
    &= \Re \sum_{j=1}^d\sum_{k=1}^d \sum_{\ell=1}^d\sum_{m=1}^d U_{jk}\overline{U_{\ell m}}~\overline{V_{jk}}V_{\ell m}.
\end{align*}
Then, taking the expected value
\begin{align*}
    \ev{|\langle V,U\rangle|^2} &= \Re \sum_{j=1}^d\sum_{k=1}^d \sum_{\ell=1}^d\sum_{m=1}^d U_{jk}\overline{U_{\ell m}}\ev{\overline{V_{jk}}V_{\ell m}}\\
    &\overset{*}{=} \Re \sum_{j=1}^d\sum_{k=1}^d \sum_{\ell=1}^d\sum_{m=1}^d U_{jk}\overline{U_{\ell m}}\frac{\delta_{j\ell}\delta_{km}}{d},
\end{align*}
where (*) uses Equation 10 of \cite{mcclean2018barren}. See \cite{collins2003moments,collins2006integration,puchala2017symbolic} for more details on this formula. Therefore, simplifying the last equation, we obtain
\begin{align*}
    \ev{|\langle V,U\rangle|^2} = \frac{1}{d}\sum_{j=1}^d\sum_{k=1}^d |U_{jk}|^2= \frac{1}{d}\|U\|_F^2= 1.
\end{align*}
Finally, the result follows via Markov's inequality.
\end{proof}

Lemma~\ref{lem:concentration} implies that
\begin{align*}
    \pr{\frac{1}{d}|\langle V,U\rangle|>\frac{1}{\sqrt{d}}}\leq \frac{1}{d}
\end{align*}
and $d=2^n$, so we have a height that goes to zero exponentially fast with respect to increasing $n$ and the probability that the objective is higher than that height goes to zero exponentially fast with increasing $n$. Thus, for a sufficiently high number of qubits ($n\approx 9$), the landscape at random initialization is prohibitively flat if $\Vct(\theta)$ is Haar distributed.

{\bf Beyond the Haar distributed assumption.} As we said, the lemma is not an if and only if result, meaning that the assumption of the Haar distribution may not hold and still the cost may have the barren plateaus property. In this context, the work \cite{mcclean2018barren} considers a relaxed assumption and shows that parameterized quantum circuits with similar objectives to ours satisfy the barren plateaus property if only the random parameterized circuit is a unitary 2-design, which means it only agrees with the Haar distribution up to the second moment. Further, the work \cite{cerezo2021cost} shows that for parameterized circuits such as $\Vct$, the number of CNOTs, $L$, and the type of objective function, determine whether the objective satisfies the barren plateaus property or not. If the objective is global, that is, if it requires measurement of all $n$ qubits, then it satisfies the barren plateaus property for all $L$. If the objective is local, that is, if it requires measurement of only one qubit, then it satisfies the barren plateaus property for $L\geo (n)$. Our objective in~\eqref{theproblem} is global, which means that despite it not being Haar-distributed, it satisfies the barren plateaus property.

\subsection{Proof of Lemma~\ref{lem:ratio}}
\label{app:ratio}

For $k\in[n]$, define
\begin{align*}
    Y_k = \frac{X_1^2+\cdots+X_k^2}{X_1^2+\cdots+X_n^2}.
\end{align*}
Then $Y_k\in[0,1]$ almost surely and, for all $t\in(0,1)$,
\begin{align*}
    &\pr{Y_k \leq t}\\
    &= \pr{\frac{\left(X_1^2+\cdots+X_k^2\right)/k}{\left(X_{k+1}^2+\cdots+X_n^2\right)/(n-k)}\leq \frac{t/k}{(1-t)/(n-k)}}\\
    &= I_t\left(\frac{k}{2},\frac{n-k}{2}\right)
\end{align*}
where $I$ is the regularized incomplete beta function and the second equality comes from the CDF of the F-distribution. See Chapter 27 of \cite{johnson1995continuous}. This proves that $Y_k\sim \bd(k/2,(n-k)/2)$. Thus, from Chapter 25 of \cite{johnson1995continuous},
\begin{align*}
    \ev{Y_k} = \frac{k}{n}, \, \textrm{and, for all $p\in\N$,}\,\, \ev{Y_k^{p+1}} = \frac{k+2p}{n+2p}\ev{Y_k^p}.
\end{align*}
Now we are ready to compute the variance of $Y(a)$. First,
\begin{align*}
    \ev{Y(a)} = \ev{Y_1}\sum_i a_i = \frac{1}{n}\sum_i a_i.
\end{align*}
Next, given $x\in\R^n$, observe that
\begin{align*}
    \left(\sum_i a_ix_i\right)^2 &= \sum_{i,j} a_ia_jx_ix_j\\
    &= \frac{1}{2}\sum_{i,j} a_ia_j\left[(x_i+x_j)^2-x_i^2-x_j^2\right].
\end{align*}
Expanding the right-hand side, 
\begin{align*}
   \textrm{RHS} &=\frac{1}{2}\sum_{i,j} a_ia_j(x_i+x_j)^2-\frac{1}{2}\sum_j a_j \sum_i a_i x_i^2\\
    &\hspace{3cm}-\frac{1}{2}\sum_i a_i \sum_j a_j x_j^2\\
    &= \frac{1}{2}\sum_{i\neq j} a_i a_j (x_i+x_j)^2+2\sum_i a_i^2 x_i^2\\
    &\hspace{3cm}-\left(\sum_i a_i\right)\left(\sum_i a_i x_i^2\right).
\end{align*}
So,
\begin{align*}
    &\ev{Y(a)^2}\\
    &= \frac{1}{2}\ev{Y_2^2}\sum_{i\neq j} a_i a_j+\ev{Y_1^2}\left(2\sum_i a_i^2-\left(\sum_i a_i\right)^2\right)\\
    &= \frac{4}{n(n+2)}\sum_{i\neq j} a_i a_j+\frac{3}{n(n+2)}\left(\sum_i a_i^2-\sum_{i\neq j}a_i a_j\right)\\
    &= \frac{1}{n(n+2)}\sum_{i\neq j} a_i a_j+\frac{3}{n(n+2)}\sum_i a_i^2.
\end{align*}
Finally,
\begin{align*}
    & \var\left(Y(a)\right) = \ev{Y(a)^2}-\ev{Y(a)}^2
    \\ &=\! \frac{2(n\!-\!1)}{n^2(n\!+\!2)}\sum_i a_i^2 - \frac{2}{n^2(n\!+\!2)}\sum_{i\neq j}a_i a_j  \leq \frac{4(n\!-\!1)}{n^2(n\!+\!2)}\sum_i a_i^2
\end{align*}
where we use Young's inequality in the last step.

\subsection{Results for the algorithm S\&S-1 }\label{app.ss1}

We show here that \eqref{approx} does not hold if $\text{Rng}(\Vct)=\SU(d)$, where $\text{Rng}$ denotes the range of a function. However, the dimension of the real manifolds $\text{Rng}(\Vct)$ and $\SU(d)$ are $3n+4L$ and $4^n-1$ respectively. So, for $L\eeo(n)$, the dimension of $\text{Rng}(\Vct)$ is on the order of log of the dimension of $\SU(d)$. This leaves hope of proving that the sketch-and-solve technique works. We leave a proof or counter-example as a future research direction.

To be precise, if we set $X\in\C^{d\times m}$ by normalizing the columns of $\Omega$, then we would like to bound
\begin{align*}
    \pr{\sup_{V\in \text{Rng}(\Vct)}\bigg|\frac{1}{m}\Re\langle VX,UX\rangle-\frac{1}{d}\Re\langle V,U\rangle\bigg|>t}.
\end{align*}

First, let's see what we can say with the supremum outside the probability measure.

\begin{theorem}
\label{thm:outsidesup}
Let $U\in\U(d)$. Let $x_1,\ldots,x_m$ be uniformly sampled complex unit vectors. Then, for all $\delta\in(0,1)$,
\begin{align*}
    &\sup_{V\in\U(d)}P\Bigg(\bigg|\frac{1}{m}\sum_{k=1}^m \Re\langle Vx_k,Ux_k\rangle-\frac{1}{d}\Re\langle V,U\rangle \bigg|\\
    &\hspace{3cm}>\sqrt{\frac{4(d-1)}{md(d+2)\delta}}\Bigg)\leq \delta.
\end{align*}
\end{theorem}
\begin{proof}
Given $t>0$,
\begin{align*}
     &\sup_{V\in\U(d)}\pr{\bigg|\frac{1}{m}\sum_{k=1}^m \Re\langle Vx_k,Ux_k\rangle-\frac{1}{d}\Re\langle V,U\rangle \bigg|>t}\\
     &\leq \frac{1}{t^2}\sup_{V\in\U(d)} \var\left(\frac{1}{m}\sum_{k=1}^m \Re\langle Vx_k,Ux_k\rangle\right)\\
     &= \frac{1}{mt^2}\sup_{V\in\U(d)} \var\left(\Re\langle Vx_1,Ux_1\rangle\right)\leq \frac{4(d-1)}{mt^2d(d+2)}
\end{align*}
where the first inequality follows by Chebyshev's inequality, and the second inequality follows from Lemma~\ref{lem:unitbatch}. Setting the right-hand side to $\delta$ and solving for $t$ gives the result.
\end{proof}

If we move the supremum in Theorem~\ref{thm:outsidesup} inside of the probability measure, then we have the following lower bound.

\begin{theorem}
Let $U\in\U(d)$. Let $x_1,\ldots,x_m$ be uniformly sampled complex unit vectors. Then
\begin{align*}
    &\mathbb{P}\Bigg(\sup_{V\in\U(d)}\bigg|\frac{1}{m}\sum_{k=1}^m \Re\langle Vx_k,Ux_k\rangle-\frac{1}{d}\Re\langle V,U\rangle \bigg|\\
    &\hspace{3cm}>\frac{2(d-m)}{d}\Bigg)=1.
\end{align*}
\end{theorem}
\begin{proof}
Let $[Q_1~Q_2][R_1\tran~0]\tran$ be the QR decomposition of $[x_1|\cdots|x_m]$. Set $V'=U(Q_1Q_1\dagg-Q_2Q_2\dagg)$. It can be verified that $V'\in\U(d)$, $\langle V'x_k,Ux_k\rangle=1$ for all $k\in\{1,\ldots,m\}$, and $\langle V',U\rangle = 2m-d$. So, for $t>0$,
\begin{align*}
    &\pr{\sup_{V\in\U(d)} \bigg|\frac{1}{m}\sum_{k=1}^m \Re\langle Vx_k,Ux_k\rangle-\frac{1}{d}\Re\langle V,U\rangle\bigg|>t}\\
    &\ge \pr{\bigg|\frac{1}{m}\sum_{k=1}^m \Re\langle V'x_k,Ux_k\rangle-\frac{1}{d}\Re\langle V,U\rangle\bigg|>t}\\
    &= \pr{\bigg|1-\frac{2m-d}{d}\bigg|>t}= \pr{\frac{2(d-m)}{d}>t}.
\end{align*}
\end{proof}

Thus, in order to potentially prove that the sketch-and-solve technique works, we really do need to restrict the supremum to be over the  range of $\Vct$.

\newpage
\bibliographystyle{IEEETran}
\bibliography{main}

\label{LastPage}
\end{document}